\documentclass[letterpaper,10pt,conference]{ieeeconf}  

\IEEEoverridecommandlockouts                              
\usepackage{mathptmx} 
\usepackage{times} 
\usepackage{amsmath} 
\usepackage{amssymb}  

\usepackage[noadjust,sort]{cite}

\let\labelindent\relax
\usepackage{enumitem}

\usepackage{subcaption}

\title{\LARGE \bf Distributed Agreement on Activity Driven Networks}

\author{
Masaki Ogura, Junpei Tagawa, and Naoki Masuda
\thanks{M. Ogura and J. Tagawa are with the Graduate School of Information Science, Nara Institute of Science and Technology, Ikoma, Nara 630-0192, Japan.
Email:  {\tt\small oguram@is.naist.jp, tagawa.junpei.sy2@is.naist.jp}}%
\thanks{N. Masuda is with the Department of Engineering Mathematics, University of Bristol, Clifton, Bristol BS8 1UB, United Kingdom. 
Email:  {\tt\small naoki.masuda@bristol.ac.uk}}%
}

\newtheorem{definition}{Definition}[section]

\newtheorem{proposition}[definition]{Proposition}
\newtheorem{theorem}[definition]{Theorem}

\usepackage{chngcntr}
\usepackage{apptools}
\AtAppendix{\counterwithin{definition}{section}}

\newcommand{\norm}[1]{\lVert #1 \rVert}

\newcommand{\abs}[1]{| #1 |}

\DeclareSymbolFont{bbold}{U}{bbold}{m}{n}
\DeclareSymbolFontAlphabet{\mathbbold}{bbold}
\newcommand{\onev}{\mathbbold{1}}

\usepackage{calrsfs}
\DeclareMathAlphabet{\pazocal}{OMS}{zplm}{m}{n}
\renewcommand{\mathcal}[1]{\pazocal{#1}}

\newcommand{\ER}{Erd\H{o}s-R\'enyi }

\usepackage{color}

\usepackage{mathtools}
\mathtoolsset{centercolon}

\newcommand{\afterequation}{\vskip 3pt}

\begin{document}

\maketitle
\thispagestyle{empty}
\pagestyle{empty}

\begin{abstract}
In this paper, we investigate asymptotic properties of a consensus protocol
taking place in a class of temporal (i.e., time-varying) networks called the
\emph{activity driven network}. We first show that a standard methodology
provides us with an estimate of the convergence rate toward the consensus, in
terms of the eigenvalues of a matrix whose computational cost grows
exponentially fast in the number of nodes in the network. To overcome this
difficulty, we then derive alternative bounds involving the eigenvalues of a
matrix that is easy to compute. Our analysis covers the regimes of 1) sparse
networks and 2) fast-switching networks. We numerically confirm our theoretical
results by numerical simulations.
\end{abstract}

\section{Introduction}

The distributed agreement problem on networks, often referred to as a consensus problem~\cite{Olfati-Saber2007}, is an important problem in network science and engineering, with applications in multi-agent coordination~\cite{Ren2005}, distributed computing~\cite{Xiao2004}, distributed sensor networks~\cite{Cortes2005}, and power systems~\cite{Dorfler2012}. The most fundamental problem in this context is convergence analysis, where we judge if a given consensus protocol allows nodes in a network to eventually achieve a consensus. Reflecting the aforementioned wide range of applications, available in the literature are various different consensus protocols and their analysis (see, e.g., \cite{Hatano2005,Moreau2005,Tahbaz-Salehi2008,Abaid2011,Touri2014} and references therein).

In this paper, we are interested in consensus protocols where connectivity of a network changes over time (referred to as \emph{temporal networks} in network science~\cite{Holme2012,Holme2015b,Masuda2016b}), since many realistic consensus dynamics take place in temporal rather than static networks~\cite{Boyd2006c,Holme2014b}. In fact, we find a plethora of results on consensus protocols on temporal networks in the literature of systems and control theory. However, many of them still employ temporal networks whose connectivity is generated from the Erd\H{o}s-R\'enyi random graph model~\cite{Hatano2005,Zhou2009,Mousavi2017,Hale2017}, which fails to replicate major characteristics of empirical temporal networks~(see, e.g.,~\cite{Masuda2016b}).

In order to fill this gap, in this paper we focus on consensus protocols taking place in a class of temporal networks called the \emph{activity driven network}~\cite{Perra2012}. Although the activity driven network is relatively simple, the model can reproduce an arbitrary degree distribution. Properties of activity driven networks have been investigated; examples include structural properties~\cite{Perra2012,Starnini2013b}, equilibrium properties of random walks \cite{Perra2012a,Ribeiro2013}, and spreading dynamics~\cite{Perra2012,Rizzo2014,Speidel2016a}.

We investigate a continuous-time consensus protocol taking place in activity driven networks. We first show that a standard methodology~\cite{Hatano2005,Hale2017} provides us with estimates of the rate of convergence toward the consensus in terms of the eigenvalues of a matrix whose computational complexity grows exponentially fast with the number of nodes in the network. To overcome this limitation, we then show that the rate of convergence can be upper-bounded by a function of the second largest eigenvalue of a matrix that is computationally feasible to calculate. The proof is based on a rigorous estimate of the matrix exponential of Laplacian matrices, and does not rely on approximations such as truncations of matrix exponentials~\cite{Hale2017}.

Our analysis focuses on the following two regimes. We first study the consensus protocol under the assumption that snapshots (i.e., networks at a given time point) of the temporal networks are sparse. This assumption is justified by the fact that several empirical temporal networks have sparse snapshots~\cite{Holme2015b}. We then focus on the regime where switching of network connectivity is sufficiently fast. Although asymptotic properties of consensus dynamics on fast-switching networks are analyzed in~\cite{Belykh2004a,Stilwell2005,Frasca2008}, these papers do not give the rate of convergence toward the consensus. Our analysis allows us not only to analyze asymptotic stability of the consensus protocol, but also to bound its rate of convergence.

This paper is organized as follows. After introducing mathematical notations, in Section~\ref{sec:problemSetting} we introduce the activity driven network and formulate the consensus protocol taking place therein. In Sections~\ref{sec:sparse} and~\ref{sec:fast}, we present our bounds on the convergence rate of the consensus protocol in the regimes of sparse networks and fast-switching networks, respectively. We numerically confirm our theoretical results in Section~\ref{sec:numerical}.

\subsection{Mathematical Preliminaries}

We let $I_n$ denote the identity matrix of dimension $n$ and~$O_{n,m}$ denote the $n\times m$ zero matrix. We denote by~$\onev_n$ the $n$\nobreakdash-dimensional vector whose elements are all one. We let $J_n$ denote the $n\times n$ matrix whose elements are all one. For a matrix $M$, let $M^\top$ denote the transpose of $M$. If $M$ is symmetric, we denote the (real) eigenvalues of $M$ by~$\lambda_1(M) \leq \cdots \leq \lambda_n(M)$. For a vector $x$, its Euclidean norm is denoted by~$\norm{x}$. The probability of an event is denoted by~$P(\cdot)$. For an integrable random variable $X$, we let $E[X]$ denote its expectation.

An undirected network is a pair~$G = (\mathcal V, \mathcal E)$, where $\mathcal V = \{1, \dotsc, n\}$ is the set of nodes, and~$\mathcal E$ is the set of edges, consisting of distinct and unordered pairs~$\{i, j\}$ for $i, j\in \mathcal V$. We say that nodes $i$ and~$j$ are adjacent if $\{i, j\} \in \mathcal E$. The adjacency matrix~$A\in \mathbb{R}^{n\times n}$ of~$G$ is defined as the $\{0, 1\}$\nobreakdash-matrix whose $(i,j)$ entry is one if and only if nodes $i$ and $j$ are adjacent. The degree of node $i$, denoted by~$d_i$, is defined as the number of nodes adjacent to node $i$. We let $D$ denote the $n\times n$ diagonal matrix having the diagonals~$d_1$, $\dotsc$, $d_n$. The (combinatorial) Laplacian matrix of $G$ is defined by~$L = D-A$.

\section{Problem Setting} \label{sec:problemSetting}

In this section, we first introduce the activity driven network~\cite{Perra2012}. Then, we formulate a consensus process taking place in activity driven networks. We finally discuss computational difficulty in analyzing the asymptotic behavior of the consensus process.

\subsection{Activity Driven Network}

The activity driven network~\cite{Perra2012} is a discrete-time model of temporal networks defined as follows:

\begin{definition}[\cite{Perra2012}]\label{defn:ADM}
For each $i \in \mathcal V$, let $a_i$ be a positive constant less than or equal to $1$. We call $a_i$ the \emph{activity rate} of node $i$. Also, let $m$ be a positive integer less than or equal to $n-1$. The \emph{activity driven network} is defined as an independent and identically distributed sequence $ \{G_k\}_{k= 0}^\infty$ of undirected graphs created by the following procedure (see Fig.~\ref{fig:adn} for a schematic illustration):
\begin{enumerate}
\item At each time $k\geq 0$, each node $i$ becomes ``activated'' with
probability~$a_i$ independently of other nodes.

\item An activated node (if any) randomly chooses $m$ other nodes and creates
$m$ (undirected) edges between them. These edges are discarded before time
$k+1$.

\item The above procedure is repeated over a range of $k$. 
\end{enumerate}
\end{definition}

\begin{figure}[tb]
\centering \includegraphics[clip,trim={2cm 4.2cm 1.2cm
6.8cm},width=.95\linewidth]{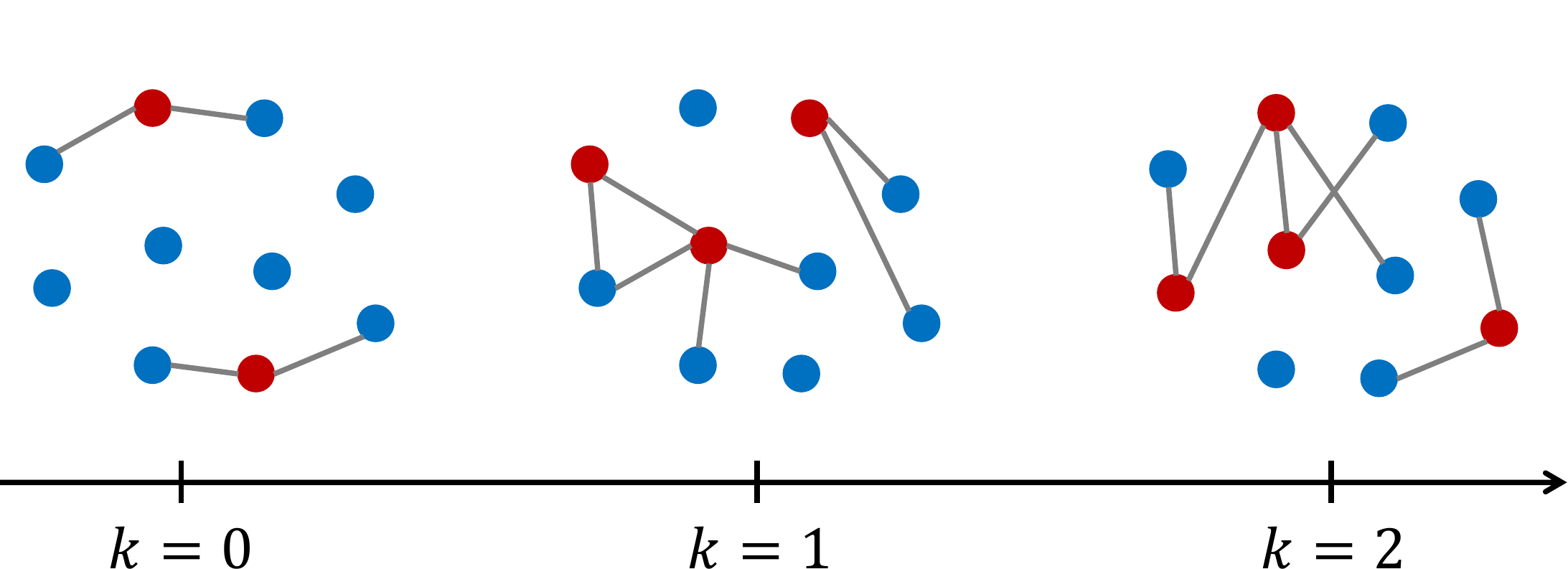} \caption{Activity driven
network. We set $n=10$ and $m=2$. Red circles represent active nodes. }
\label{fig:adn}
\end{figure}

\subsection{Consensus Protocol}\label{sec:cp}

Following the problem setup of~\cite{Hatano2005,Masuda2013a,Hale2017}, we assume that the network, on which nodes seek for consensus, is a continuous-time temporal network whose connectivity changes with a period $\Delta t>0$. Instead of assuming that the snapshots of the temporal network is the Erd\H{o}s-R\'enyi random graph model~\cite{Hatano2005,Hale2017}, we use snapshots created by the activity driven network. Specifically, let us consider the continuous-time temporal network $\{G(t)\}_{t\geq 0}$ given by
\begin{equation}\label{eq:tempNet}
G(t) = G_k,\quad k\Delta t\leq t < (k+1)\Delta t, 
\end{equation}
where $\{G_k\}_{k = 0}^\infty$ is the activity driven network.

We examine the following consensus protocol in continuous-time (see, e.g.,
\cite{Hatano2005,Olfati-Saber2007}): 
\begin{equation}\label{eq:cProtocol}
dx_i/dt = \sum_{j \in \mathcal N_i(t)} (x_j(t) - x_i(t)), \ t\geq 0, 
\end{equation}
where $x_i(t)$ is a real number and~$\mathcal N_i(t)$ denotes the set of nodes adjacent to node $i$ in network~$G(t)$. One can rewrite \eqref{eq:cProtocol} as
\begin{equation*}
{dx}/{dt} = -L(t)x(t),
\end{equation*}
where $x = [x_1\ \cdots \ x_n]^\top$, and~$L(t)$ is the Laplacian matrix of network~$G(t)$. Following Refs.~\cite{Hatano2005,Masuda2013a,Hale2017}, we focus on the dynamics of the periodic samples
\begin{equation*}
z_k = x(k \,\Delta t),\quad k=0, 1, 2, \dotsc,
\end{equation*}
of the state of the continuous-time consensus protocol~\eqref{eq:cProtocol}. The
periodic samples obey the discrete-time dynamics given by
\begin{equation}\label{eq:disc}
z_{k+1} = e^{- \Delta t L_k } z_k, 
\end{equation}
where $L_k$ denotes the Laplacian matrix of $G_k$.

Let us define the ``consensus space'' $\mathcal A$ (see, e.g., \cite{Hatano2005,Hale2017}) as the subspace of $\mathbb{R}^n$ spanned by the all-one vector~$\onev_n$. Let $\Pi$ be the orthogonal projection in $\mathbb{R}^n$ onto $\mathcal A^\perp$, the hyperplane in $\mathbb{R}^n$ orthogonal to $\mathcal A$. Then, the projection~$\Pi z_k$ allows us to measure how far from consensus the nodal states are at time $k$. In terms of this quantity, we can prove the following proposition that shows convergence of the consensus protocol and gives an expression of the convergence rate. The proof of the proposition is almost the same as \cite[Proposition~III.1]{Hatano2005} for the case of temporal networks with \ER snapshots and, therefore, is omitted.

\begin{proposition}\label{prop:conv}
Let $\epsilon > 0$. The consensus protocol~\eqref{eq:disc} satisfies
\begin{equation}\label{eq:convergence}
P \biggl( \sup_{k \geq K} \norm{\Pi z_k}^2 \geq \epsilon \biggr) 
\leq 
\epsilon^{-1}{\norm{\Pi z_0}^2} \lambda_{n-1} (E[e^{-2\Delta t L_k}])^K
\end{equation}
for all $K\geq 0$ and~$z_0\in\mathbb{R}^n$.
\end{proposition}

\subsection{Computational Difficulty} \label{sec:diffc}

Although Proposition~\ref{prop:conv} guarantees convergence toward the consensus with probability 1 and gives us an upper-bound on the convergence rate, the bound is not necessarily computationally tractable. This is because the matrix~$E[e^{-2\Delta t L_k}]$ is hard to calculate for large networks. To illustrate the difficulty, let us count how many different snapshots are possible in the activity driven network. For simplicity, we here ignore the cases where a pair of activated nodes choose each other as their neighbors. If $n_a$ nodes are activated at a specific time, there exist $\binom{n-1}{m}^{n_a}$ different ways in which the activated nodes choose their neighbor. Since there are $\binom{n}{n_a}$ different combinations of $n_a$ activated nodes, the total number of possible snapshots in the activity driven network equals
\begin{equation}\label{eq:Lnm}
N_{n,m} = \sum_{n_a=0}^n \binom{n}{n_a} \binom{n-1}{m}^{n_a}
= \biggl[1 + \binom{n-1}{m}\biggr]^n, 
\end{equation}
which grows exponentially in the number of nodes~$n$. For example, when $n=10$, we obtain $N_{n, m} \geq 10^{10}$ whenever $m\neq 9$. This implies that, to compute the expectation~$E[e^{-2\Delta t L_k}]$ appearing in \eqref{eq:convergence}, we need to compute no less than $10^{10}$ matrix exponentials of Laplacian matrices, which is demanding. We further remark that, for the case where snapshots of the temporal network are \ER random graphs, the authors in~\cite{Hale2017} evaluated the same matrix $E[e^{-2\Delta t L_k}]$ (for the Laplacian matrix of \ER random graphs) using a  truncation of the series expansion of matrix exponentials. However, the truncation error was not clarified in the paper.

\section{Sparse Networks}\label{sec:sparse}

In this section, we study asymptotic behavior of the consensus protocol in the regime of sparse networks where nodes have small activity rates. Under this regime, we will bound the convergence rate in terms of the eigenvalues of a matrix that can be easily calculated.

When nodes in the network have small activity rates, we expect that at most one node is activated in each snapshot~$G_k$ with a high probability. Under this regime, we can effectively approximate the original activity driven network by the temporal network~$\{G'_k\}_{k= 0}^\infty$ defined as follows:

\begin{enumerate}[label=\arabic*{$'$})]
\item At each discrete time step $k\geq 0$, at most one node becomes activated.
Specifically, for each $i\in \mathcal V$, node~$i$ is activated with probability
$a_i$ (therefore, no node is activated with probability $1-\sum_{i=1}^na_i$),
independently of other nodes. 

\item An activated node (if any) randomly chooses $m$ other nodes and creates
$m$ (undirected) edges between them. These edges are discarded before
time~$k+1$.

\item The above procedure is repeated over a range of $k$. 
\end{enumerate}

As we did in \eqref{eq:tempNet}, let us define the continuous-time temporal
network~$\{G'(t)\}_{t\geq 0}$ by
\begin{equation*}
G'(t) = G_k',\quad k\Delta t \leq t < (k+1)\Delta t,
\end{equation*}
and let $L'(t)$ denote the Laplacian matrix of the network~$G'(t)$. As in
Section~\ref{sec:cp}, we consider the consensus protocol $d x'/dt =
-L'(t)x'(t)$, and focus on the periodic samples
\begin{equation}\label{eq:disc'}
z'_{k+1} = e^{-\Delta t L'_k} z'_k,\quad z'_k = x'(k \Delta t),  
\end{equation}
where $L'_k$ is the Laplacian matrix of $G'_k$. In the same way as in
Proposition~\ref{prop:conv}, we can show that the consensus
protocol~\eqref{eq:disc'} satisfies
\begin{equation}\label{eq:convergence'}
P \biggl( \sup_{k \geq K} \norm{\Pi z_k'}^2 \geq \epsilon \biggr) 
\leq 
\epsilon^{-1}{\norm{\Pi z_0'}^2} \lambda_{n-1} (E[e^{-2\Delta t L'_k}])^K
\end{equation}
for all $\epsilon>0$, $K\geq 0$, and~$z'_0\in\mathbb{R}^n$. The following
proposition shows that,  unlike in the case of the original activity driven
network, we can easily compute the matrix~$E[e^{-2\Delta t L_k'}]$ on the
right-hand side of \eqref{eq:convergence'}:

\begin{proposition}\label{prop:convRate'}
For each $i \in \mathcal V$, define $M_i \in \mathbb{R}^{n\times n}$
by
\begin{equation}\label{eq:M:diag}
\begin{aligned}
[M_i]_{kk} = 
\begin{cases}
1 - \dfrac{m(1 - e^{-(m+1)2\Delta t})}{m+1},

\quad \mbox{if $k = i$,}
\\
1 - \dfrac{(m^2-1)e^{-2\Delta t} + e^{-(m+1)T} - m^2}{(m+1)(n-1)}, 
\quad\mbox{otherwise,}
\end{cases}
\end{aligned}
\end{equation}
for all $k$, and 
\begin{equation}\label{eq:M:nondiag}
\begin{aligned}
[M_i]_{k\ell} = 
\begin{cases}
\dfrac{(m-1)(1-e^{-(m+1)2\Delta t})}{(m+1)(n-2)(n-1)},
\quad
\mbox{if $k=i$ or $\ell = i$,}
\vspace{.5em} \\
\dfrac{m + e^{-(m+1)2\Delta t} - (m+1)e^{-2\Delta t}}{(m+1)(n-1)},
\quad
\mbox{otherwise,}
\end{cases}
\end{aligned}
\end{equation}
for all distinct pairs $(k, \ell)$. Then, it holds that
\begin{equation}\label{eq:desiredd}
E[e^{-2\Delta t L_k'}] = \biggl(1 -\sum_{i=1}^n a_i \biggr)I_n + \sum_{i=1}^n a_i M_i. 
\end{equation}
\afterequation
\end{proposition}

\begin{proof}
Let us write $T = 2\Delta t$. Fix $k\geq 0$ and define the events
\begin{equation}\label{eq:def:events}
\begin{aligned}
\mathcal A_i &= \{\mbox{node~$i$ is activated in $G_k'$}\}, \ (i = 1, \dotsc, n),
\\
\mathcal B &= \{\mbox{no node is activated in $G_k'$}\}. 
\end{aligned}
\end{equation}
The probability of the events~$\mathcal A_i$ and~$\mathcal B$ are $a_i$ and
$1-\sum_{i=1}^n a_i$, respectively. Therefore, we have
\begin{equation}\label{eq:obsvand...}
E[e^{-T L_k'}] 
=
\biggl(1 - \sum_{i=1}^n a_i\biggr) E[e^{-T L_k'} \mid \mathcal B] +
\sum_{i=1}^n a_i E[e^{-T L_k'} \mid \mathcal A_i]. 
\end{equation}
When no node is activated (i.e., $\mathcal B$ occurs), we trivially have $L_k' =
O$ and, therefore, 
\begin{equation}\label{eq:notappearing...}
E[e^{-T L_k'} \mid \mathcal B] = I.
\end{equation}
When node $i$ is
activated, we obtain
\begin{equation}\label{eq:appearing...}
E[e^{- T L_k'} \mid \mathcal A_i] 
= 
\sum_{\mathcal N \in \mathfrak N_i} \frac{1}{\binom{n-1}{m}} e^{-T L_{i, \mathcal N} }, 
\end{equation}
where $\mathfrak N_i$ denotes the set of all the subsets of $\mathcal V\backslash \{i\}$ having exactly $m$ elements (recall that an activated node randomly chooses $m$ nodes according to the uniform distribution), and~$L_{i, \mathcal N}$ denotes the Laplacian matrix of the network having nodes $\mathcal V$ and edges $\mathcal E_{i, \mathcal N} = \{ \{i, j\} \}_{j\in\mathcal N}$. Furthermore, as we will show later in this proof, the matrix exponential appearing in the right-hand side of \eqref{eq:appearing...} is given by
\begin{equation}\label{eq:itsufficies:diag}
\begin{multlined}
[e^{-T L_{i, \mathcal N}}]_{kk} = 
\\ 
\begin{cases}
1 - \dfrac{m(1 - e^{-(m+1)T})}{m+1}, 
& 
\mbox{if $k=i$,}
\vspace{.5em} \\
e^{-T} + \dfrac{m + e^{-(m+1)T} -(m+1)e^{-T}}{m(m+1)}, 
&
\mbox{if $k \in \mathcal N$,} 
\vspace{.5em} \\
1, 
&
\mbox{otherwise,}
\end{cases}
\end{multlined}
\end{equation}
for all $k$, and 
\begin{equation}\label{eq:itsufficies:nondiag}
\begin{aligned}
&[e^{-T L_{i, \mathcal N}}]_{k\ell} = 
\\
& 
\begin{cases}
\dfrac{1-e^{-(m+1)T}}{m+1}, 
&  
\mbox{if $\{k, \ell\} \in \mathcal E_{i, \mathcal N},$} 
\vspace{.5em} \\
\dfrac{m + e^{-(m+1)T} -(m+1)e^{-T}}{m(m+1)}, 
&
\mbox{if $k\in\mathcal N$ and $\ell \in \mathcal N$,} 
\vspace{.5em} \\
0, 
&
\mbox{otherwise, }
\end{cases}
\end{aligned}
\end{equation}
for all distinct pairs $(k, \ell)$. Equations \eqref{eq:appearing...}--\eqref{eq:itsufficies:nondiag} prove $E[e^{- T L_k'} \mid \mathcal A_i] = M_i$. Hence, equations \eqref{eq:obsvand...}--\eqref{eq:appearing...} imply that desired equation~\eqref{eq:desiredd} holds true.

In the rest of the proof, we show that equations~\eqref{eq:itsufficies:diag} and~\eqref{eq:itsufficies:nondiag} hold true. Without loss of generality, we assume $i = 1$ and~$\mathcal N = \{2, 3, \dotsc, m+1\}$. It is sufficient to show that
\begin{equation}\label{eq:itsufficies}
\begin{aligned}
&e^{-T L_{i, \mathcal N}} = 
\\
&\!\!\!\!\begin{bmatrix}
\frac{me^{-(m+1)T} + 1}{m+1} 
& 
\frac{1-e^{-(m+1)T}}{m+1}\onev^\top_{m} & 0_{n-m-1}^\top
\\
\frac{1-e^{-(m+1)T}}{m+1}\onev_{m}  
&
\!\!\! e^{-T}I_m + 
\frac{m + e^{-(m+1)T} -(m+1)e^{-T}}{m(m+1)} J_m \!\!\! & O_{m,n-m-1}
\\
0_{n-m-1} & O_{n-m-1,m} & I_{n-m-1}
\end{bmatrix}\!. 
\end{aligned}
\end{equation}
As proved in the appendix, for every $k\geq 1$, it holds that
\begin{equation}\label{eq:inductino}
\!\!(L_{i, \mathcal N})^k \!\!=\!\! \begin{bmatrix}
m(m+1)^{k-1} & -(m+1)^{k-1}\onev^\top_m & 0^\top_{n-m-1}
\\
-(m+1)^{k-1}\onev_m & 
\!\!I_m + \frac{(m+1)^{k-1} - 1}{m} J_m \!\! & O_{m, n-m-1}
\\
0_m & O_{m, n-m-1} & O_{n-m-1}
\end{bmatrix}\!. 
\end{equation}
Now, let us partition the matrix~$e^{-T{L_{i, \mathcal N}}}$ as
\begin{equation*}
e^{-T{L_{i, \mathcal N}}	} = 
\begin{bmatrix}
N_{11} & N_{12} & N_{13}
\\
N_{21} & N_{22} & N_{23}
\\
N_{31} & N_{32} & N_{33}
\end{bmatrix}
\end{equation*}
according to the block structure in \eqref{eq:inductino}.
Equation~\eqref{eq:inductino} implies that matrices $N_{13}$, $N_{23}$,
$N_{31}$, and~$N_{32}$ are zero, and~$N_{33}$ is the identity matrix.
Furthermore, using the definition of matrix exponential and
\eqref{eq:inductino}, we obtain
\begin{equation*}
\begin{aligned}
N_{11}
&= 1 + \sum_{k=1}^{\infty} \frac{m(m+1)^{k-1}(-T)^k}{k!}
\\
&= 1 + \frac{m}{m+1}\left (e^{-(m+1)T} - 1\right)
\\
&=
\frac{me^{-(m+1)T} + 1}{m+1}, 
\\
N_{12} 
& = \sum_{k=1}^{\infty} \frac{-(m+1)^{k-1}\onev^\top (-T)^k}{k!}
\\
&= -\frac{1}{m+1} (e^{-(m+1)T} - 1)\onev^\top, 
\\
N_{22} &
= 
I_m + \sum_{k=1}^{\infty} \frac{I_m + \frac{(m+1)^{k-1} - 1}{m} J_m}{k!} (-T)^k
\\
&=
I_m + (e^{-T}-1)I_m + \frac{e^{-(m+1)T}-1}{m(m+1)}J_m - \frac{e^{-T}-1}{m}J_m, 
\end{aligned}
\end{equation*}
which proves \eqref{eq:itsufficies}. This completes the proof of the
proposition.
\end{proof}

We now present our first main result, which gives an efficient method for evaluating the convergence rate of the consensus protocol~\eqref{eq:disc'}. Specifically, the following theorem allows us to evaluate the convergence rate in terms of the eigenvalues of a linear combination of matrices~$M_i$ given by \eqref{eq:M:diag} and \eqref{eq:M:nondiag}, overcoming the computational difficulty illustrated in~Section~\ref{sec:diffc}.

\begin{theorem}\label{thm:slow}
Let $\epsilon > 0$ be arbitrary. In the consensus protocol~\eqref{eq:disc'}, the
probability $P( \sup_{k \geq K} \norm{\Pi z_k'} \geq \epsilon)$ converges to
zero exponentially fast as $K\to\infty$ with a decay rate less than or equal to
\begin{equation}\label{eq:bound:sp}
\gamma_{\textrm{\,SP}} = 1-\sum_{i=1}^n a_i  + \lambda_{n-1}\biggl(\sum_{i=1}^n a_i M_i\biggr).
\end{equation}
\afterequation
\end{theorem}

\begin{proof}
Directly follows from the inequality~\eqref{eq:convergence'} and the expression
of the matrix $E[e^{-2\Delta t L'_k}]$ in \eqref{eq:desiredd}.
\end{proof}

\section{Fast-Switching Networks} \label{sec:fast}	

In this section, we analyze asymptotic behavior of the consensus protocol~\eqref{eq:disc} over the activity driven network under the regime of fast switching~\cite{Belykh2004a,Stilwell2005,Frasca2008} (i.e., small $\Delta t$). Also in this regime, we can evaluate the rate of convergence toward the consensus in terms of the eigenvalues of a matrix that is easy to compute.

Instead of directly studying the consensus protocol~\eqref{eq:disc}, we first introduce an alternative discrete-time temporal network, denoted by~$ \{G''_k\}_{k= 0}^\infty$, produced in the following procedure:

\begin{enumerate}[label=\arabic*{$''$})]
\item At each time $k\geq 0$, each node $i$ becomes activated with
probability~$a_i$ independently of other nodes.

\item \label{item:forSubordinance} Let $S$ be the set of activated nodes. If
$\abs{S} \geq 2$, then we randomly choose one and only one node, say $i$, from
the set~$S$ with probability $\beta_{S, i} \geq 0$ and keep it activated, while
inactivating all the other nodes.

\item An activated node (if any) randomly chooses $m$ other nodes and creates
$m$ (undirected) edges between them.   These edges are discarded before
time~$k+1$.

\item The above procedure is repeated over a range of $k$. 
\end{enumerate}

The only difference of the above procedure from that for the activity driven network is step~\ref{item:forSubordinance}, where we randomly inactivate all but one nodes. This additional step guarantees that each snapshot $G_k''$ is a ``subgraph'' of snapshot $G_k$ in the activity driven network. Due to this property, we expect that consensus will be reached in the original activity driven network faster than in~$\{G_k'' \}_{k=0}^\infty$. Consistent with this intuition, we obtain the following theorem that allows us to efficiently evaluate the convergence rate toward the consensus in the activity driven network, under the regime of fast switching:

\begin{theorem}
Let $\epsilon > 0$. Assume that $\Delta t > 0$ is sufficiently small. Then, in
the consensus protocol~\eqref{eq:disc}, the probability~$P( \sup_{k \geq K}
\norm{\Pi z_k} \geq \epsilon)$ converges to zero exponentially fast as
$K\to\infty$ with a decay rate less than or equal to
\begin{equation*}
\gamma_{\textrm{\,FS}} =  1-\sum_{i=1}^n b_i  + \lambda_{n-1}\biggl(\sum_{i=1}^n b_i M_i\biggr), 
\end{equation*}
where the matrices $M_i$ are given in \eqref{eq:M:diag} and \eqref{eq:M:nondiag} and 
the constants $b_i$ are defined by 
$b_i = p_{\{i\}} + \sum_{S\in 2^\mathcal V,\, \abs{S}\geq 2} p_S \beta_{S, i}$
with $p_S = (\prod_{i \in S} a_i ) ( \prod_{j \in
\mathcal V \backslash S} (1-a_j))$ for every subset $S$ of $\mathcal V$.
\end{theorem}

\begin{proof}
Let us write $T = 2\Delta t$. We claim that, if $T$ is sufficiently small, then
\begin{equation}\label{eq:Tsmall}
\lambda_{n-1} (E[e^{- T L_k }]) 
\leq 
\lambda_{n-1} (E[e^{- T L_k''}]). 
\end{equation}
We shall temporarily assume that this claim holds true.  Fix $k\geq 0$, and define the events $\mathcal A_i$ and~$\mathcal B_i$ in the same way as in \eqref{eq:def:events}. Then, the probability of events~$\mathcal A_i$ and~$\mathcal B$ in network~$G''_k$ is $b_i$ and~$1-\sum_{i=1}^n b_i$, respectively, by the definition of the constants $b_i$. Therefore, in the same way as in the proof of Proposition~\ref{prop:convRate'}, we can show $E[e^{-T L_k''}] = (1 -\sum_{i=1}^n b_i )I + \sum_{i=1}^n b_i M_i$. From this equation and \eqref{eq:Tsmall}, we can immediately prove the theorem using Proposition~\ref{prop:conv}.

In the rest of the proof of this theorem, we prove inequality~\eqref{eq:Tsmall} under the assumption that $T$ (therefore, $\Delta t$) is sufficiently small. Similar to our analysis in Section~\ref{sec:diffc}, there exist networks~$G_{(\ell)}$ ($\ell = 1, \dotsc, N_{n,m}$) such that their Laplacian matrices $L_{(\ell)}$ satisfy
\begin{equation}\label{eq:exp}
E[e^{-T L_k }] = N_{n,m}^{-1} \sum_{\ell=1}^{N_{n,m}} e^{-T L_{(\ell)}},
\end{equation}
where $N_{n,m}$ is defined in \eqref{eq:Lnm}. Similarly, by the definition of
network $G''_k$, there exists another family of networks $G''_{(\ell)}$
($\ell=1, \dotsc, N_{n, m}$) such that $G''_{(\ell)}$ is a subgraph of
$G_{(\ell)}$ for all $\ell$, and their Laplacian matrices $L_{(\ell)}''$ satisfy
\begin{equation}\label{eq:exp''}
E[e^{- T L_k''}] = N_{n,m}^{-1} \sum_{\ell=1}^{N_{n,m}} e^{-T L_{(\ell)}''}. 
\end{equation}

Let $A_{(\ell)}$ and~$A_{(\ell)}''$ be the adjacency matrices of networks~$G_{(\ell)}$ and  $G_{(\ell)}''$, respectively. Consider the weighted network~$\Delta G$ having the adjacency matrix~$ \sum_{\ell=1}^{N_{n, m}} (A_{(\ell)} - A''_{(\ell)})$. Since $G_{(\ell)}''$ is a subgraph of $G_{(\ell)}$ for all $\ell$, edges in network~$\Delta G$ have nonnegative weights. In fact, for any pair~$(i, j)$ of distinct nodes, the edge~$\{i, j\}$ has a positive weight in $\Delta G$. This is because, if the weight of an edge in $\Delta G$ were equal to zero, then the edge would have to appear in networks $\{G_{(\ell)}\}_{\ell=1}^{N_{n, m}}$ as many times as in $\{G_{(\ell)}''\}_{\ell=1}^{N_{n, m}}$, which contradicts the definition of network~$G_k''$ (specifically, step~2'')). Hence, the Laplacian matrix~$\Delta L$ of network~$\Delta G$ is positive semidefinite with its null space being spanned by~$\onev_n$, i.e., we have
\begin{equation}\label{eq:>0}
x^\top \Delta L x > 0
\end{equation}
for all $x \in \mathcal A^\perp \backslash \{0\}$. 

Now, let us define the matrices $L = \sum_{\ell=1}^{N_{n, m}} L_{(\ell)}$ and $L'' = \sum_{\ell=1}^{N_{n, m}} L_{(\ell)}''.$ Since $\Delta L = L-L''$, inequality~\eqref{eq:>0} implies $x^\top L x > x^\top L'' x$ for all $x \in \mathcal A^\perp \backslash \{0\}$. Therefore, the Rayleigh-Ritz theorem~\cite{Horn1990} shows $\lambda_{n-1}(L) > \lambda_{n-1}(L'')$. Let us take an arbitrary $c_1>0$ such that $\lambda_{n-1}(L) \geq c_1N_{n,m}+ \lambda_{n-1}(L'')$. Then, for any $T>0$, we have
\begin{equation}\label{eq:wehave}
\lambda_{n-1}\left(I-\frac{TL}{N_{n,m}}\right) \leq 
\lambda_{n-1}\left(I-\frac{TL''}{N_{n,m}}\right) - c_1T. 
\end{equation}
Define the matrices 
\begin{equation}\label{eq:VV''}
\begin{aligned}
V 
&= \left(I-\frac{TL}{N_{n,m}}\right) - E[e^{-TL_k}], 
\\
V'' 
&= \left(I-\frac{TL''}{N_{n,m}}\right) - E[e^{-TL_k''}].
\end{aligned}
\end{equation}
Then, Weyl's inequality and \eqref{eq:wehave} imply
\begin{equation}\label{eq:wehave''}
\lambda_{n-1}(E[e^{-TL_k}]) + \lambda_1(V)  \leq 
\lambda_{n-1}(E[e^{-TL_k''}]) + \lambda_n(V'')  - c_1T.
\end{equation}
Since equations \eqref{eq:exp}, \eqref{eq:exp''}, and \eqref{eq:VV''} imply that both $V$ and~$V''$ are of order~$O(T^2)$, there exist positive constants $c_2$ and $\delta$ such that, if $T<\delta$, then $\abs{\lambda_1(V)} < c_2T^2$ and $\abs{\lambda_n(V'')} < c_2 T^2$. These inequalities and \eqref{eq:wehave''} imply that, if $T < \delta$, then $\lambda_{n-1}(E[e^{-TL_k}]) \leq \lambda_{n-1}(E[e^{-TL_k''}]) - c_1T + 2c_2T^2$. Hence, inequality~\eqref{eq:Tsmall} holds true for $T\leq \min(\delta, c_1 c_2^{-1}/2)$, as desired. This completes the proof of the theorem.
\end{proof}

\section{Numerical Simulations} \label{sec:numerical}

In this section, we numerically illustrate our results obtained in the previous
sections. Due to limitations of the space, we focus on the regime of sparse
networks (studied in Section~\ref{sec:sparse}). We first consider an activity driven network with $n=10$ nodes. We independently draw the activity rates of
nodes from the uniform distribution on $[0, 1/100]$. We set $m = 4$, $\Delta t =
2$, and $\epsilon = 0.3$. We empirically compute the probabilities~$P(\sup_{k
\geq K} \norm{\Pi z_k'}^2 \geq \epsilon)$ for $0\leq K<100$ using $10^4$ sample
paths of the discrete-time consensus dynamics~\eqref{eq:disc'} with a randomly
chosen initial condition $z'_0$. The empirical probabilities are shown in
Fig.~\ref{fig:}.\subref{fig:1}. We see that our theoretical bound on the decay
rate~$\gamma_{\textrm{\,SP}}$ (represented by the blue dashed line) is in a good
fit with the actual decay rate.

We then run another simulation for a larger activity driven network having $n=50$ nodes and activity rates drawn from the uniform distribution on $[0, 1/500]$. We set $m = 15$, $\Delta t = 3$, and $\epsilon = 0.1$. We empirically compute the probabilities~$P(\sup_{k \geq K} \norm{\Pi z_k'}^2 \geq \epsilon)$ (shown in Fig.~\ref{fig:}.\subref{fig:fig2}), and estimate the decay rate to be $0.9789$. The relative error of our theoretical upper bound~$\gamma_{\textrm{\,SP}} = 0.9855$ from the actual decay rate is less than $0.7\%$, which confirms a high accuracy of our bound.

\newcommand{\mywidth}{.7975\linewidth}

\begin{figure}[tb]
\vspace{1mm}
\centering
\begin{subfigure}[b]{\linewidth}
\centering \includegraphics[width=\mywidth]{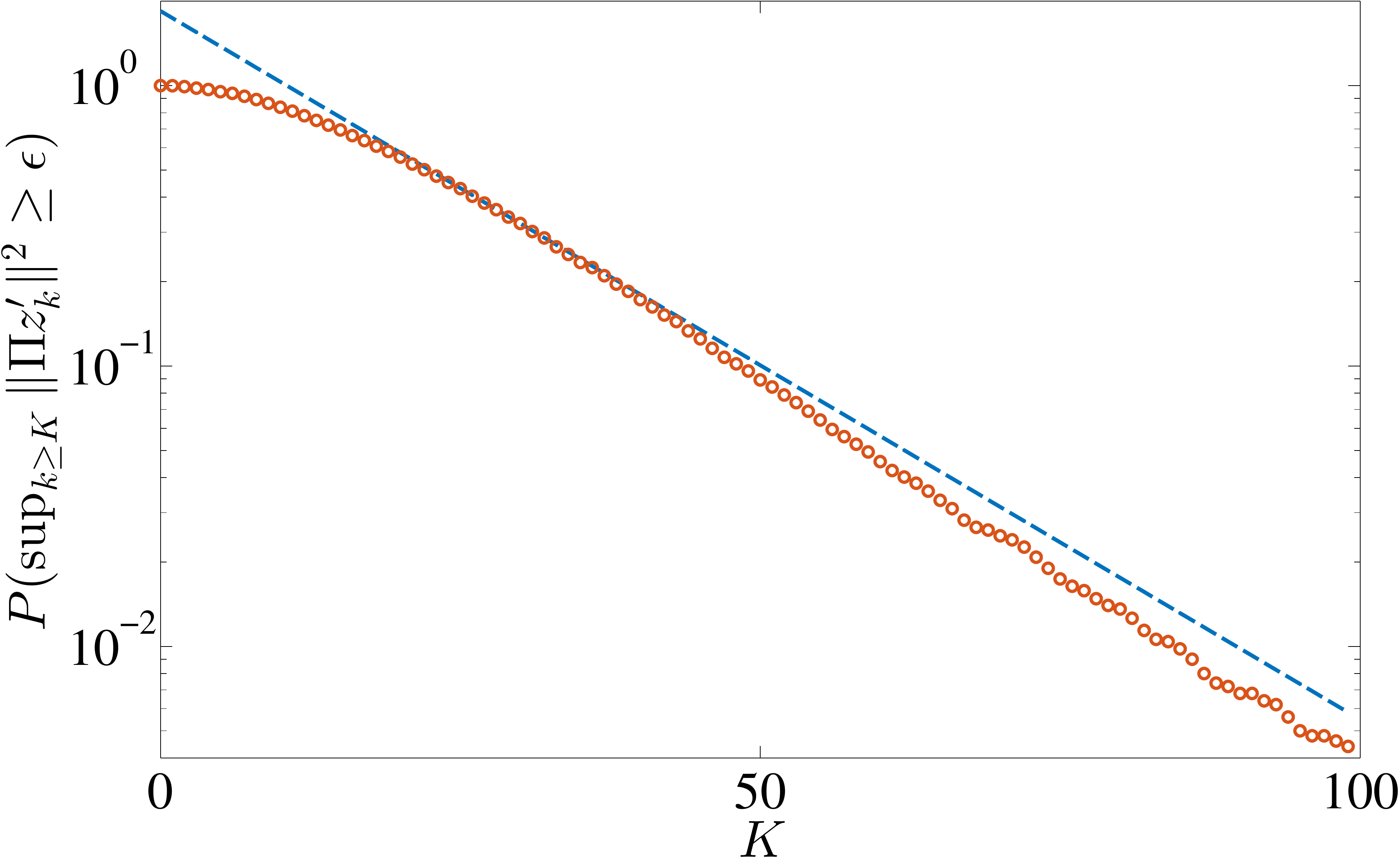}
\caption{Numerical results for an activity driven network with $n=10$ nodes. Activity rates are drawn from the uniform distribution on $[0, 1/100]$. We set $m=4$, $\Delta t = 2$, and $\epsilon = 0.3$.}
\label{fig:1}
\end{subfigure}
\vspace{2mm}
\\
\begin{subfigure}[b]{\linewidth}
\centering \includegraphics[width=\mywidth]{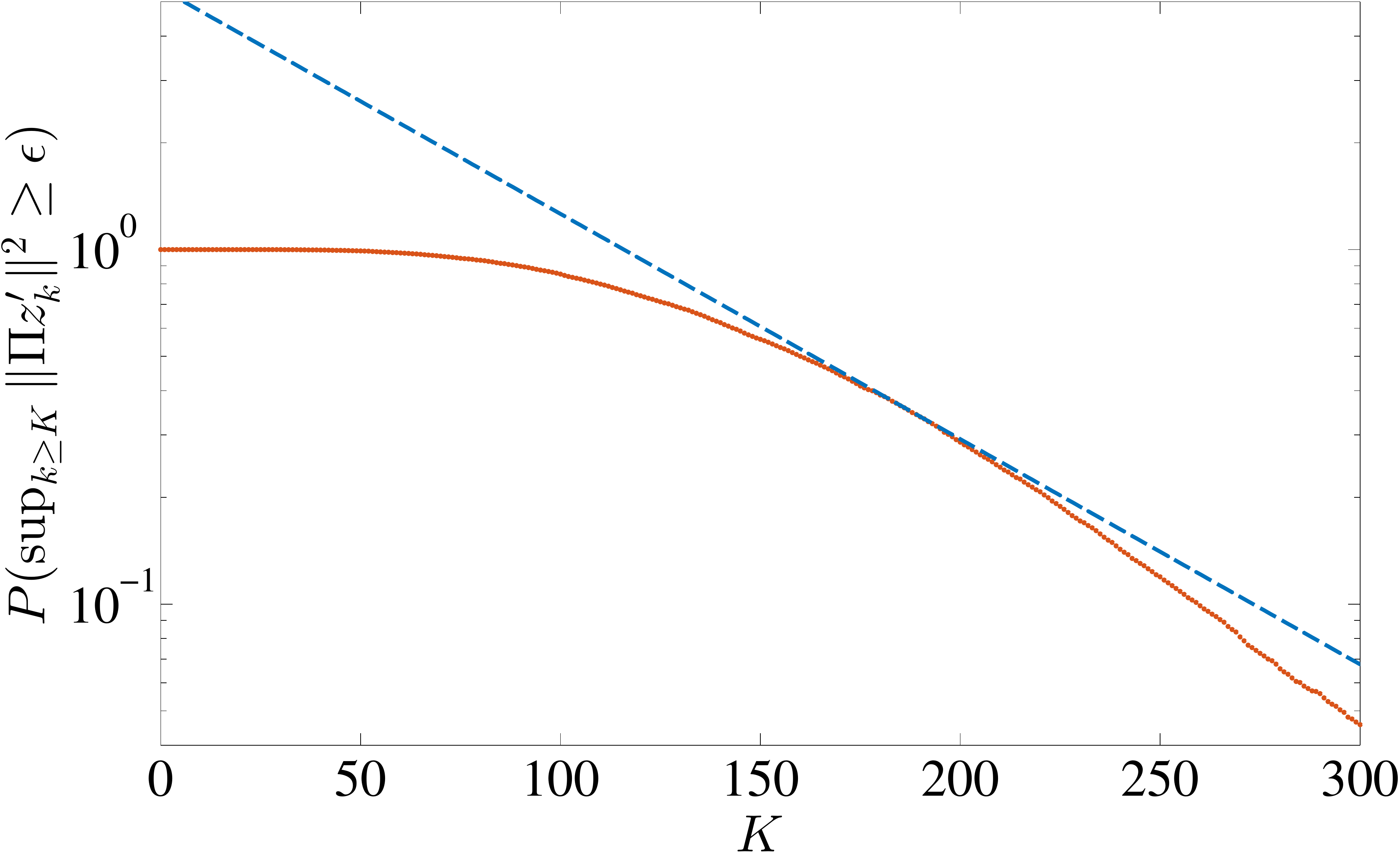}
\caption{Numerical results for an activity driven network $n=50$ nodes. Activity rates are drawn from the uniform distribution on $[0, 1/500]$. We set $m=15$, $\Delta t = 5$, and $\epsilon = 0.1$.}
\label{fig:fig2}
\end{subfigure}
\caption{The probability of having $\sup_{k \geq K}\norm{\Pi z_k'}^2 \geq \epsilon$ (circles).
The dashed line
shows geometric sequences having scale factor $\gamma_{\textrm{\,SP}}$ (given in \eqref{eq:bound:sp}).}
\label{fig:}
\end{figure}

\section{Conclusion}

We have studied consensus dynamics taking place in activity driven networks. We have presented computationally efficient frameworks for evaluating the convergence rate toward consensus, in terms of the eigenvalues of matrices that are easy to compute. Our analyses cover the cases of sparse networks and fast-switching networks, and avoid computations of the matrix exponentials of exponentially many Laplacian matrices required by an existing method. We have numerically confirmed our theoretical results via simulations.

\appendix

\renewcommand{\thesection}{A}

In this appendix, we prove equation~\eqref{eq:inductino} using an induction on
$k$. Since $L$ is the Laplacian matrix of the network having nodes $\mathcal V$
and edges $\{ \{1, j\} \}_{2\leq j\leq m+1}$, we have
\begin{equation}\label{eq:L}
L = \begin{bmatrix}
m & -\onev^\top_m & 0^\top
\\
-\onev_m &I_m & O
\\
0 & O & O
\end{bmatrix}, 
\end{equation}
showing that \eqref{eq:inductino} holds true for $k = 1$. Now, we
assume that \eqref{eq:inductino} holds true for $k = k_0$ and compute $L^{k_0+1}$.
Let us partition matrix $L^{k_0+1}$ as
\begin{equation*}
L^{k_0+1} = 
\begin{bmatrix}
M_{11} & M_{12} & M_{13}
\\
M_{21} & M_{22} & M_{23}
\\
M_{31} & M_{32} & M_{33}
\end{bmatrix} 
\end{equation*}
according to the block structure of $L$ in \eqref{eq:L}. From our assumption, it
follows that
\begin{equation*}
\begin{aligned}
M_{11} &= 
m(m+1)^{k_0-1} \cdot m +(-(m+1)^{k_0-1}\onev_m^\top)(-\onev_m)
\\
&=
m^2(m+1)^{k_0-1} + m(m+1)^{k_0-1}
\\
&=
m(m+1)^{k_0}, 
\\
M_{12} 
&= 
m(m+1)^{k_0-1}  (-\onev_m^\top) +(-(m+1)^{k_0-1}\onev_m^\top)I_m
\\
&= - (m+1)^{k_0} \onev_m, 
\\
M_{22} 
&= 
-(m+1)^{k-1}\onev_m (-\onev_m^\top) + \biggl( I_m + \frac{(m+1)^{k-1} - 1}{m} J_m\biggr)I_m
\\
&=
(m+1)^{k_0-1}J_m + I_m + \frac{(m+1)^{k_0-1} - 1}{m}J_m
\\
&=
I_m + \frac{(m+1)^{k_0} - 1}{m} J_m. 
\end{aligned}
\end{equation*}
We can also easily confirm that matrices $M_{13}$, $M_{23}$, and~$M_{33}$ are
zero. Therefore, by the symmetry of matrix~$L^{k_0+1}$, we conclude that
equation \eqref{eq:inductino} holds true for $k = k_0+1$ as well. Hence, an
induction on $k$ completes the proof of equation~\eqref{eq:inductino} for every
$k\geq 1$.


\end{document}